\RequirePackage{fix-cm}
\documentclass[smallextended]{svjour3}       
\smartqed  
\usepackage{graphicx}
\usepackage{amsmath}
\usepackage{amssymb}
\usepackage{pxfonts}
\usepackage{eucal}
\usepackage{mathrsfs}
\usepackage{pifont}
\usepackage{color}
\usepackage{mathptmx}      

\newcommand{\Raz}[1]{{\color{blue} #1}}

\usepackage{amsmath}
\usepackage{amssymb}
\usepackage{graphicx}
\usepackage{eucal}
\usepackage{mathrsfs}
\usepackage{pifont}
\usepackage{cjhebrew}

\newcommand{\bbR}{{\mathbb R}}

\newcommand{\brk}[1]{\left(#1\right)}          

\newcommand{\mymat}[1]{\begin{pmatrix} #1 \end{pmatrix}}

\newcommand{\deriv}[2]{\frac{d#1}{d#2}}
\newcommand{\dderiv}[2]{\frac{d^2#1}{d#2^2}}

\newcommand{\pd}[2]{\frac{\partial#1}{\partial#2}}
\newcommand{\pdd}[2]{\frac{\partial^2#1}{\partial#2^2}}

\newcommand{\secref}[1]{Section~\ref{#1}}
\newcommand{\figref}[1]{Figure~\ref{#1}}

\newcommand{\beq}{\begin{equation}}
\newcommand{\eeq}{\end{equation}}
\newcommand{\bsplit}{\begin{split}}
\newcommand{\esplit}{\end{split}}
\newcommand{\baligned}{\begin{aligned}}
\newcommand{\ealigned}{\end{aligned}}

\newcounter{sect}

\providecommand{\R}{\bbR}

\newcommand{\textand}{\quad\text{ and }\quad}
\newcommand{\Textand}{\qquad\text{ and }\qquad}

\newcommand{\GG}{\mathcal{G}}

\newcommand{\G}{\mathrm{G}}
\newcommand{\B}{\mathcal{B}}
\renewcommand{\S}{\mathcal{S}}
\newcommand{\vp}{\varphi}

\newcommand{\ks}{\kappa^*TS}

\newcommand{\Q}{\mathcal{Q}}
\newcommand{\W}{\mathcal{W}}

\DeclareMathOperator{\divergence}{div}

\DeclareMathOperator{\Vol}{Vol}

\newcommand{\Emb}{\operatorname{Emb}}
\newcommand{\Hom}{\operatorname{Hom}}

\newcommand{\rmin}{R_{\text{min}}}
\newcommand{\rmax}{R_{\text{max}}}

\DeclareFontFamily{U}{FBB}{}
\DeclareFontShape{U}{FBB}{m}{n}{
   <-> s * [1.0] fourier-bb
   }{}
\DeclareSymbolFont{Ufutm}{U}{FBB}{m}{n}
\DeclareSymbolFontAlphabet{\fourierbb}{Ufutm}
\AtBeginDocument{\let\mathbb\fourierbb} %

\begin{document}

\title{Continuum Dynamics on Manifolds: \\Application to Elasticity of Residually-Stressed Bodies\thanks{
RK was partially supported by the Israel-US Binational Foundation (Grant No. 2010129), by the Israel Science Foundation (Grant No. 661/13), and by a grant from the Ministry of Science, Technology and Space, Israel and the Russian Foundation for Basic Research, the Russian Federation.
RS was partially supported by 
the H. Greenhill Chair for Theoretical and Applied Mechanics and by the Pearlstone Center for Aeronautical Engineering Studies at Ben-Gurion University.
}
}


\author{Raz Kupferman \and Elihu Olami  \and Reuven Segev
}


\institute{Raz Kupferman \at
              Institute of  Mathematics \\ 
		The Hebrew University \\
		Jerusalem 91904 Israel \\
                \email{raz@math.huji.ac.il}       
           \and
           Elihu Olami \at
              Institute of  Mathematics \\ 
		The Hebrew University \\
		Jerusalem 91904 Israel \\
                \email{elikolami@gmail.com}
              \and
Reuven Segev \at
Department of Mechanical Engineering \\
Gen-Gurion University of the Negev \\
Beer-Sheva 84105 Israel                \\
                \email{rsegev@post.bgu.ac.il}
}

\date{Received: date / Accepted: date}

\maketitle

\begin{abstract}
This paper is concerned with the dynamics of continua on differentiable manifolds. 
We present a covariant derivation of equations of motion, viewing motion as a curve in an infinite-dimensional Banach space of embeddings of a body manifold in a space manifold. Our main application is the motion of residually-stressed  elastic bodies; residual stress results from a geometric incompatibility between body and space manifolds. 
We then study a particular example of elastic vibrations of a two-dimensional curved annulus embedded in a sphere.

\end{abstract}

\section{Introduction}

In the past decade, there has been a renewed interest in the mechanics of residually-stressed elastic materials. This recent activity encompasses a wide scope of branches, ranging from the natural sciences (e.g., \cite{FSDM05,LM09,AAESK12}), through engineering applications (e.g., \cite{KES07}) and up to pure mathematical questions. In the latter context, we mention the derivation of dimensionally-reduced plate, shell and rod models \cite{LP10,KS14,KM14}, and homogenization theories for topological defects \cite{KM15,KM16,KMR16}.   

Mathematically, certain residually-stressed elastic bodies may be modeled as smooth manifolds endowed with a Riemannian metric; the metric represents local equilibrium distances between neighboring material elements. A configuration is an embedding of the body manifold into the ambient Euclidean space. The elastic energy associated with a configuration is a measure of mismatch between the intrinsic metric of the body and its ``actual" metric---the pullback of the Euclidean metric by the configuration. The property of being residually-stressed is  a \emph{geometric incompatibility}, reflected,  in the traditional Euclidean settings, by the non-flatness of the intrinsic material metric. Incompatible elasticity has a longstanding history, starting with the pioneering work of Kondo \cite{Kon55}, Nye \cite{Nye53}, Bilby \cite{BBS55} and Kr\"oner \cite{Kro81}. The above mentioned recent work extends significantly the scope of applications,  provides a wealth of novel analytical tools, and raises new questions.

Almost all the existing work on residually-stressed bodies is restricted to static problems of hyperelasticity: one postulates the existence of an energy function and the equilibrium configuration is a minimizer of that energy. In contrast, very little work exists on elastodynamics residually-stressed bodies. 

Thus far, there have been several dominant approaches for covariant theories of elastodynamics:
\begin{enumerate}

\item Balance laws for extensive observables, such as mass, momentum and energy: these laws 
 are postulated along with invariances under certain types of spatial diffeomorphisms; 
see for example the Green-Rivlin theorem \cite{GR64}  and its covariant generalization by  Marsden \cite{MH83}. 
Under certain regularity assumptions, the balance equations give rise to local differential transport equations; see Marsden and Hughes \cite{MH83} and Yavari, Marsden and Ortiz \cite{YMO06}.  

\item  Field theoretic approaches: one postulates the existence of an energy function, or a Lagrangian density function $\mathcal{W}$, which depends on both intrinsic and ``actual" metrics of the body. The dynamical solution (which is a motion) is the minimizer of the corresponding energy functional; see e.g.  Ebin \cite{Ebi93}, Marsden and Hughes \cite{MH83} and Yavari  and Marsden \cite{YM12}.

\item Dynamics are viewed as statics in $4$-dimensional space-time. See for example Appleby and Kadianakis \cite{AK86}.
\end{enumerate}

In the formulations based on the Green-Rivlin theorem, as presented in Marsden and Hughes \cite{MH83}, the form of the energy balance has to be assumed a-priori. Moreover, one has to assume the existence of an elastic energy. Such an approach is somewhat inconsistent with the traditional approach in continuum mechanics, according to which balance laws have to be formulated independently of constitutive theory. This approach also restricts the theory to hyperelastic systems. The same comment applies to field theories based on a predefined form of the Lagrangian. 

In this paper, we present a global approach to continuum dynamics, with particular relevance to elastodynamics. Our main application is geometrically incompatible elastic media. Our formulation is a generalization of Newton's classical mechanics to the infinite-dimensional continuum context. It applies to a rather general class of problems, including non-conservative systems and singular systems (e.g., forces and stresses are allowed to be measure-valued). 

Writing  the laws of dynamics requires a specification of  a geometric model of space-time. Here, space-time has a particularly simple structure: a Cartesian product $\S\times I$ of an $m$-dimensional space manifold $\S$ and a time interval  $I\subset\R$. Thus, given a compact $d$-dimensional body manifold $\B$, a natural choice for the configuration space, which we denote  by $\Q$ , is the space $\Q=\Emb^1(\B,\S)$ of $C^1$ embeddings of $\B$ in $\S$. A motion of the body $\B$ in $\S$ is a curve $\vp:I\to\Q$. 

As $\S$ is generally not a linear space, neither is $\Q$. However, $\Q$ turns out to be an (infinite dimensional) Banach manifold. The tangent space of $\Q$ at a configuration $\kappa$  is identified with the Banachable space of vector fields along $\kappa$ 
\[
T_\kappa\Q\simeq C^r(\kappa^*T\S)\simeq\{\xi:\B\to T\S\,|\,\pi_\S\circ\xi=\kappa\},
\]
 where $\pi_\S:T\S\to \S$ is the tangent bundle projection.  Consequently, a  generalized velocity at a configuration $\kappa$ is modeled by a vector field along $\kappa$, $v\in C^1(\ks)\simeq T_\kappa\Q$, whereas  a generalized force is modeled by a linear functional  $f\in (C^1(\ks))^*\simeq T_\kappa^*\Q$. The action $f(v)$ is interpreted as virtual power or virtual work.

The dynamics of a system are induced by a Riemannian metric $\GG$ and a connection $\nabla^\Q$ on $T\Q$. The metric assigns to a generalized velocity the corresponding generalized momentum and the connection enables one to view the rate of change of the momentum as an element of $T^*\Q$. Thus, the dynamic law, which is a generalization of Newton's second law, states that the total generalized force is equal to the covariant derivative of the momentum with respect to time.
 
As shown in \cite{Seg86}, since the topology of $\Q$ takes into account first derivatives, so do the forces in $T^*\Q$;  a generalized force $f\in T^*_\kappa\Q$ may be represented as a function depending linearly on generalized velocities and their first derivatives. In other words, there exists a non-unique stress measure $\sigma$  satisfying the principle of virtual work, 
\[
f(v)=\sigma(j^1(v))
\]
for all generalized velocities $v$. Here, $j^1$ is the jet extension mapping of velocity fields, which is the invariant representation of the value of a vector field along  with its first derivative (the local representation of the jet extension is presented below).
Using the dual of the jet mapping, one can write
\[
f=j^{1*}\sigma=\sigma\circ j^1.
\]
The representation of forces by stresses is a pure mathematical result based on the Hahn-Banach theorem and the Riesz representation theorem of functionals by measures. In particular, it does not involve any physical notions such as balance of forces, equilibrium, external forces and internal forces. 

Traditional formulations of the dynamic law for continuous bodies are formulated in terms of the resultants of the external forces, which are integrals of force densities over their domain of definition (e.g., \cite[p.~170]{Tru91}). Such formulations are not possible in the geometric setting of manifolds, where forces are defined only in the context of their actions on virtual velocity fields and where ``rigid'' velocity fields are not defined.

We can now make the traditional assumption that the total force $f_T$ acting on a body is the sum of the total external force $f_E$ and the total internal force  $f_I$. Typically, the external force is given by a loading section $\Q\to T^*\Q$ and the internal force is represented by a stress $\sigma$, which, in turn, is usually determined by a constitutive relation. Thus,
\[
f_T=f_E - j^{_1*}\sigma.
\]
 The dynamics law proposed in \secref{sec:covariant_continuum_dynamics} is
\[
f_E - j^{1*}\sigma=\frac{DP}{dt},
\]
where $P=\GG(V,\cdot)$ is the generalized momentum, that is, the dual pairing of the velocity $V$; ${DP}/{dt}$ is the covariant derivative of $P$ along the motion. This law is equivalent to the principle of virtual work
\[
f_E(w) - \sigma(j^1 w)=\frac{DP}{dt}(w)
\]
(\emph{cf.} \cite[p.~168]{MH83}).

As a main application for this theory, we investigate the dynamics of residually-stressed hyperelastic materials. We consider a quadratic hyperelastic constitutive model with free boundary. We write the equations of motion in explicit form, yielding a nonlinear wave equation. This example demonstrates one of the peculiarities of continuum mechanics on manifolds. On a manifold, one cannot disassociate the derivative of a vector field from its value. Consequently, the stress field contains, in addition to a term dual to the derivative of the virtual velocity field, a term dual to the virtual velocity field itself. This term, sometimes referred to as the \emph{self force}, vanishes in our example if and only if the spatial metric $g$ is Euclidian  (see \cite{Cap89}). 

As a particular system, we consider the case where $\B$ and $\S$ are two-dimensional, azimuthally symmetric annuli of different constant curvatures. Recently, such systems were studied experimentally by Aharoni et al. \cite{AKMMS16}.
We present numerical calculations displaying nonlinear waves for the case of a spherical annulus embedded in a sphere of different radius. 

The structure of the paper is as follows:
We start \secref{sec:mathframework} with a brief description of classical mechanics in a covariant setting. 
In \secref{sec:Geometric}, we present the geometric structure of the configuration space $\Q$, and introduce the representation of forces by stresses in both singular and smooth settings. 
In \secref{sec:covariant_continuum_dynamics} we formulate Newton's second law for continuum dynamics. To this end, one needs  a metric and a connection for $\Q$; these are defined in \secref{sec:metric_and_connection} following Eliasson \cite{Eli67}, under the assumptions that a metric and a connection are given on the space manifold $\S$ and that  $\B$ is endowed with a mass density, or a volume form. 
In \secref{sc13} we introduce the constitutive theory. We give special attention to the hyperelastic case, for which we derive explicit expressions in local coordinates. 
\secref{sec:MooneyRivlin} is devoted to a quadratic hyperelastic constitutive model with a free boundary. In \secref{sec:example} we focus on the particular case of an azimuthally symmetric annulus embedded in a sphere and present numerical calculations.

\section{Mathematical framework}\label{sec:mathframework}

Our first goal is to  present a global covariant setting for continuum mechanics, based on a geometric characterization of the infinite-dimensional configuration space.  As a prelude, we reformulate the classical Newtonian mechanics of particle systems in a general, yet fairly simple, covariant form (see Abraham and Marsden \cite{AM87} for a covariant Hamiltonian approach to mechanics).  As mentioned above, our approach is based on the assumption that space-time has the structure of a Cartesian product; in particular, points is space have an invariant meaning independent of time.

\subsection{Covariant description of classical mechanics}

In classical mechanics, the \emph{configuration space} is a smooth $d$-dimensional manifold, which we denote by $\S$. A point in $\S$ represents the positions of all the point particles of the system. A \emph{virtual displacement} at $p\in \S$ is an elements of $T_p\S$, i.e., a tangent vector at $p$. A \emph{force} at $p\in \S$ is an element of $T_p^*\S$, i.e., a cotangent vector at $p$. The action of a force $f\in T_p^*\S$ on a virtual displacement $v\in T_p\S$ yield a scalar, $f(v)$, called a \emph{virtual power}.

A \emph{motion} of the system is a smooth curve $\vp:I  \to \S$
in the configuration space where $I$ is a time interval.
The \emph{velocity} associated with the motion $\vp$ is a map $v:I\to \vp^*T\S$, defined by 
\[
v=\deriv{\vp}{t}
\] 
We adopt here the standard notation whereby $\vp^* T\S$ is the pullback of the vector bundle $T\S$ by $\vp$; 
$\vp^* T\S$ is a vector bundle over $I$, with the fiber $(\vp^*T\S)_t$ identified with the fiber $T_{\vp(t)}\S$.

In order to define the \emph{acceleration}, i.e., in order to differentiate the velocity $v$, we need a connection $\nabla^\S$ on $\S$. The acceleration is then given by
\[
a(t)=\frac{Dv}{dt}=\nabla^\S_vv.
\]
Given a local coordinate system for $\S$,  the connection is represented by Christoffel symbols $\Gamma_{ij}^k$, which are functions on $\R^d$. Let $(\vp^1,\dots,\vp^d): I\to\R^d$ denote the local representative of the motion. Then, the velocity and the acceleration take the respective forms
\[
v^i(t)=\deriv{\vp^i}{t}
\Textand
a^i(t)=\dderiv{\vp^i}{t}+\Gamma^i_{jk}(\vp(t))\deriv{\vp^j}{t}\deriv{\vp^k}{t}.
\]
Here the indexes range between $1$ and $d$, and the Einstein summation convention is assumed. 

Let $F:\S\to T^*\S$ be a \emph{force field}, i.e., a section of the cotangent bundle (a one-form), assigning a force to every configuration. Newton's law states that the total force at the current configuration equals the time derivative of the momentum, or in the case of constant mass, to the product of mass and acceleration. 

In a geometric setting, equating force with acceleration is meaningless, as the acceleration is a tangent vector, whereas the force is a cotangent vector.
To obtain a pairing between the tangent and the cotangent bundles, a Riemannian metric $g$ on $T\S$ is needed. 
Then, the momentum $P:I\to \vp^*T^*\S$ is defined by
\[
P(t) = g_{\vp(t)}(v(t),\cdot).
\]
Newton's equation of motion reads
\[
\frac{DP}{dt}(t)=F(\vp(t)).
\]   
In order to differentiate the momentum, we need a connection on $T^*\S$. Such a connection is canonically induced by the connection on $T\S$. If the metric $g$ does not depend on time and the connection $\nabla^\S$ is metrically-consistent, that is, $\nabla^\S g=0$, then
\[
\frac{DP}{dt}(t) = g_{\vp(t)}\brk{a(t),\cdot}.
\]
In coordinates, Newton's equation  reads
\[
g_{li}(\vp(t))\brk{\dderiv{\vp^i}{t}+\Gamma^i_{jk}(\vp(t))\deriv{\vp^j}{t}\deriv{\vp^k}{t}}=F_l(\vp(t)), \qquad 1\leq l\leq d.
\]
Note that the masses of the particles are incorporated in the metric $g$. 

Take for example a particle of mass $m$ moving in $\S = \R^3$. 
The pairing between the tangent and cotangent bundles is induced by the Euclidean metric, $g_{ij} = m\,\delta_{ij}$ and the (Euclidean flat) connection is given by $\Gamma^i_{jk}= 0$,
leading to the classical ``$F=ma$" equation,
\[
m\, \delta_{ij} a^j(t) =F_i(\vp(t)), \qquad 1\le i\le d.
\] 

Even though classical mechanics views the configuration space as a manifold, 
we observe that there is a one-to-one correspondence between a manifold $\S$ and the space of functions $f:\{p\}\to \S$, where $\{p\}$ is a manifold consisting of a single point. 
In other words, the configuration space can also be viewed as a space of functions between two manifolds (albeit one of which is trivial). 
This perspective is the relevant one when we turn to continuum mechanics; the point ${p}$ is replaced by a body manifold $\B$ and configurations are functions from $\B$ to $\S$. 

\subsection{Geometric setting for continuum mechanics}\label{sec:Geometric}

In this section we present the constructs needed for
a geometric formulation of continuum mechanics; see Segev~\cite{Seg86}.
The  \emph{body manifold} $\B$ is a smooth, compact, $d$-dimensional manifold with corners. 
The \emph{space manifold} $\S$ is a smooth $m$-dimensional manifold without boundary. In most classical applications, $m=3$ and $d=1,2,3$ (rod theories correspond to $d=1$, shell and membrane theories to $d=2$ and bulk theories to $d=3$).

A \emph{configuration of class $r$} is a $C^r$-embedding $\kappa:\B\to\S$ of the body manifold $\B$ in the space manifold $\S$.
The \emph{configuration space},  
\[
\Q=\Emb^r(\B,\S),
\] 
is the space of $C^r$-embeddings of the body in space. We endow $\Q$ with the subspace topology induced from the Whitney $C^r$-topology of $C^r(\B,\S)$; loosely speaking, it is the topology of uniform convergence of all derivatives up to order $r$. The configuration space $\Q$ is not a vector space, since addition is not defined on the manifold $\S$. Moreover, even in the case where the space manifold is a vector space, the set of embeddings is not a vector space. Nevertheless, $\Q$ can be given a structure of an infinite-dimensional Banach manifold---a topological space locally homeomorphic to a Banach space and equipped with a smooth structure (see e.g. Palais \cite{Pal68}, Abraham \cite{Abr64} and Eliasson \cite{Eli67}). 

The standard construction of local charts for $\Q$ relies on the existence of a connection on $\S$ (see Krikorian \cite{Kri72} for an alternative approach). We will henceforth assume that $\S$ is paracompact and therefore admits a partition of unity. Consequently, there exists a connection for $\S$. It is noted however that the differential structure for $\Q$ does not depend on the choice of a connection.  Since we will eventually take $\S$ to be a Riemannian manifold (with the canonical Levi-Civita connection), we will assume the existence of a connection $\nabla^\S$ for $\S$ from the outset.

For every $\kappa\in\Q$, there exists a neighborhood $U_\kappa \subset \Q$ of $\kappa$ and a canonical coordinate chart $\chi: C^r(\kappa^*T\S)\to U_\kappa$, where $C^r(\kappa^*T\S)$ is the Banachable space of  vector fields along $\kappa$ (with the $C^r$ topology), 
\[
C^r(\kappa^*T\S) \simeq \{v\in C^r(\B,T\S)\,\,|\,\, \pi_\S\circ v=\kappa \},
\]
and $\pi_\S:T\S\to \S$ is the projection of the tangent bundle on the base manifold. 
(By a Banachable space, we mean that $C^r(\ks)$ is a topological vector space admitting a (non-canonical) complete compatible norm.)

For $v\in C^r(\kappa^*TS)$, $\chi(v)\in C^r(\B,S)$ is given by 
\[
\chi(v)(p)=\exp(v_p)
\] 
where $\exp(v_p)$ is the value at $t=1$ of the unique geodesic $\gamma:[0,1]\to \S$ satisfying the initial condition $\dot{\gamma}(0)=v_p$ (this is where the connection $\nabla^\S$ enters). 
Thus, for every $\kappa\in \Q$, the tangent space  $T_\kappa \Q$ can be identified with the Banachable space    $C^r(\kappa^*T\S)$. As in the finite-dimensional case, the tangent bundle $T\Q=\cup_{\kappa\in\Q}T_\kappa\Q$ is  the bundle of \emph{virtual displacements}, or \emph{generalized velocities}.

In the sequel, we use the following notational convention: Spaces of $r$-times differentiable functions between two manifolds, e.g., $\B$ and $\S$, are denoted by $C^r(\B,\S)$. Spaces of $r$-times differentiable sections of vector (or more generally, fiber) bundles, e.g., $T\B$, are denoted by $C^r(T\B)$, rather than $C^r(\B,T\B)$.

A \emph{force} of grade $r$ is an element of the cotangent bundle $T^*\Q$.
Let $\kappa\in\Q$ be a configuration.  As for the finite-dimensional case, the action $f(w)$ of a force $f\in T_\kappa^*\Q$ on a virtual displacement, or generalized velocity $w\in T_\kappa\Q$ is called a \emph{virtual power}.

While the definitions thus far may seem identical to the definitions in the previous section, there exist fundamental differences between the finite- and the infinite-dimensional settings. In the finite-dimensional case, every vector space $V$ is isomorphic to its dual $V^*$. Moreover, the topology does not depend on the chosen norm (all norms are equivalent). In infinite dimension, this is no longer true; in particular, the cotangent space $T^*_\kappa\Q\simeq (C^r(\kappa^*TS))^*$ is not diffeomorphic to the tangent space $T_\kappa\Q\simeq C^r(\kappa^*TS)$. This difference has deep analytical implications. In fact, it is the origin of the introduction of stresses.

Given a configuration $\kappa\in\Q$, the cotangent space at $\kappa$ is the space of continuous linear functionals $f:C^r(\ks)\to\R$. As the topology of $\Q$ (and that of the model space $C^r(\ks)$) takes into account all the derivatives up to order $r$, so do continuous linear functionals in $T_\kappa^*\Q$; given a force $f\in T_\kappa^*\Q$ and a virtual displacement $w\in T_\kappa\Q$ at $\kappa$, their pairing $f(w)$ is a linear function of $w$ and its first $r$ derivatives. 

The mathematical construct for encoding information about the value assumed by a function along with its first $r$ derivatives at a point is that of \emph{jets} (see e.g. Saunders \cite{Sau89}). We denote by $J^r(\B,\S)$ the set consisting of points $p$ in $\B$ along with the equivalence class of all functions $\kappa:\B\to\S$ assuming at $p$ the same values in their first $r$ derivatives in some (hence, any) coordinate system. The equivalence class of a function $f$ at a point $p\in\B$ is denoted by $j^r_pf$.  The set $J^r(\B,\S)$ of equivalence classes can be given the structure of a fiber bundle over $\B$, called the $r$-th jet bundle of functions from $\B$ to $\S$.

The notion of a jet bundle is easily understood using a coordinate system. Let $X = (X^1,\dots,X^d)$ and $x = (x^1,\dots,x^m)$ be coordinate systems for $\B$ and $\S$; indexes of coordinates in $\B$ will be denoted by Greek letters, whereas indexes of coordinates in $\S$ will be denoted by Roman letters. An element of $J^r(\B,\S)$ is  represented  locally by the coordinates $X^\alpha$ of a point $X\in\R^d$ in the body manifold,  the coordinates $x^i$ of a point $x\in\R^m$ in the space manifold, and symmetric, multilinear operators,
\[
A_1: \R^d\to\R^m,
\qquad
A_2: \R^d\times\R^d\to\R^m,
\quad
\ldots
\quad
A_r:  \R^d\times\cdots\times\R^d\to\R^m.
\]
representing $x^i_{,\alpha},x^i_{,\alpha_1\alpha_2}$, etc., where commas indicate partial differentiation. Given a function $\kappa\in C^r(\B,\S)$, we denote by $j^r\kappa \in C^0(J^r(\B,\S))$ the section of the $r$-th jet bundle, returning, for every $p\in\B$ the jet defined by $\kappa$ and its first $r$ derivatives at $p$; the section $j^r\kappa$ is called the $r$-th jet prolongation of $\kappa$. In coordinates,  if  $\kappa:\B\to\S$ is represented locally by its components $(\kappa^1,\dots,\kappa^m):\R^d\to\R$, then, the local representation of its $r$-th jet prolongation is
\[
j^r\kappa(X) = (X,\kappa^i(X),\kappa^i_{,\alpha}(X),\dots,\kappa^i_{,\alpha_1\dots \alpha_r}(X)).
\]

Back to the action of a force on a virtual displacement, it follows from the Hahn-Banach theorem that given a force $f\in T_\kappa^*\Q$ , there exists a continuous linear functional $\sigma \in (C^0(J^r(\ks)))^*$ such that for every   virtual displacement $w\in T_\kappa\Q\simeq C^r(\kappa^*T\S)$, the action of a force $f$ on $w$ can be represented as
\beq
f(w) = \sigma(j^rw).
\label{stressrep}
\eeq
We call $\sigma$ a \emph{stress} at $\kappa$, and denote the space $(C^0(J^r(\ks)))^*$ of stresses at $\kappa$ by $\mathfrak{S}_\kappa$. We say that a stress $\sigma$ at $\kappa$ represents the force $f$ if equation \eqref{stressrep} holds for every choice of virtual velocity $w$. Note however, that for a given force $f$, there may be more than one stress representing it. This is referred in continuum mechanics as  \emph{static indeterminacy}. 

In general,  stresses and forces, which are continuous linear functionals on differentiable sections, may be singular. Locally, and in particular, if $\B$ can be covered by a single chart, every stress $\sigma$ can be represented by  a collection $\{\mu_i,\,\mu_i^{\alpha},\ldots,\mu_i^{\alpha_1\ldots \alpha_r}\}$ of measures  by the formula 
\[
\begin{split}
\sigma(j^rw) &= \int_\B w^i\, d\mu_i+\int_\B w^i_{,\alpha} \, d\mu^\alpha_i 
+\cdots +
\int_\B w^i_{,\alpha_1\ldots \alpha_r} \, d\mu_i^{\alpha_1 \ldots \alpha_r}.       
\end{split}
\]

We now restrict ourselves to first grade materials, i.e., $r=1$, which is a conventional modeling assumption in standard continuum mechanics, and in particular in bulk elasticity theory and in tension field theory \cite{Ste90}. Furthermore, we restrict our consideration to smooth stress measures, where $\sigma$ (at some configuration $\kappa$) is given by 
\[
\sigma(j^1w)=\int_\B S(j^1w),
\]
and $S$ is a smooth vector-valued form, which we call the \emph{variational stress density}. 
As shown in \cite{Seg02,Seg13}, we may decompose $S$ into body and surface terms as follows,
\[
\int_\B S(j^1w)= - \int_\B \divergence S(w)+\int_{\partial\B} p_\sigma S(w).
\]
Here $\divergence S$ and $p_\sigma S$ are vector-valued forms,
\[
\begin{gathered}
\divergence S \in \Gamma(\Hom(\ks,\Lambda^dT^*\B)) \\
p_\sigma S \in \Gamma(\Hom(\ks|_{\partial\B},\Lambda^{d-1}T^*\B) ),
\end{gathered}
\]
and for any vector bundle $\pi:E\to M$, $\Gamma(E)$ denotes the space of smooth sections of $E$.

In coordinates, the action of a variational stress on the jet extension of a virtual velocity is of form 
\[
S(j^1w) = (R_i w^i + S_i^\alpha w^i_{,\alpha})\, dX^1\wedge\cdots\wedge dX^d,
\]
where $R_i$ and $S_i^\alpha$ are functions of $X$. The vector-valued forms $\divergence S$ and $p_\sigma S$ are then given by
\[
\begin{split}
\divergence S(w) &= (\divergence S)_i w^i\, dX^1\wedge\cdots\wedge dX^d \\
p_\sigma S(w) &= (p_\sigma S)^\alpha_i w^i\,  \,dX^1\wedge\cdots\wedge \widehat{dX^j}\wedge\cdots\wedge dX^d ,
\end{split}
\]
where
\beq
\label{eq:divS}
(\divergence S)_i = S^\alpha_{i,\alpha} - R_i
\Textand 
(p_\sigma S)_i^\alpha =(-1)^{\alpha-1}S^\alpha_i.
\eeq
Here, the notation $\widehat{dX^j}$ indicates that the term $dX^j$ is omitted from the wedge product;
 in the expression $(-1)^{\alpha-1}S^\alpha_i$ there is no summation over $\alpha$.

The $R$-term in the action of a variational stress does not appear in conventional continuum mechanics. For continuum mechanics on non-flat manifolds, it is sometimes referred to as the self-force, see e.g. Capriz \cite{Cap89}. We will see in \secref{sec:MooneyRivlin} an example in which the $R$ term appears as a consequence of the non-flatness of the ambient space $\S$.  

Let $\kappa\in\Q$. Suppose that a force $f\in T^*_\kappa\Q$ is given by body and surface densities $b\in \Gamma(\Hom(\ks,\Lambda^dT^*\B))$ and $t\in\Gamma(\Hom(\ks|_{\partial\B},\Lambda^{d-1}T^*\partial\B))$, that is, 
\[
f(w)=
\int_\B b(w) +
\int_{\partial\B} t(w),
\]
Then, $f$ is represented by a stress at $\kappa$ with variational stress density $S$,
\[
f(w) = \int_\B S(j^1 w)
\]
if and only if 
\[
\divergence S + b=0 
\Textand  
p_\sigma S|_{\partial\B}=t.
\]
Here, $p_\sigma S|_{\partial\B}=\iota_{\partial\B}^*p_\sigma S$ where $\iota_{\partial\B}:\partial\B\to\B$ is the inclusion.

\section{Covariant continuum dynamics} 
\label{sec:covariant_continuum_dynamics}

In this section we present the equations of motion, generalizing Newton's equations to the continuum setting on manifolds. Newton's second law states that the time derivative of the momentum equals the total force acting on the body.  We  start by describing the total force acting on a body. We derive expressions for the momentum of a motion and its covariant derivative, given a general connection and a (possibly time-dependent) metric on $\Q$. With the proper notions at hand, the equations of motion are postulated. We conclude the section by constructing a metric and a connection for $Q$ in the case where $\B$ is endowed with a mass form and $\S$ is a Riemannian manifold. 

\subsection{Force, momentum and Newtons second law}
\label{sec:force}

In classical mechanics, the total force is commonly divided into two components: external forces representing ambient  fields, and internal forces representing interactions among the particles composing the system. In continuum mechanics, the force is divided into two components as well:
Let us fix a configuration $\kappa$ of the body in space. We assume that the forces are given by an external force $f_E$ and an internal force $f_I$, such that the total force $f_T\in T^*_\kappa\Q$ is given by
\[
f_T=f_E-f_I.
\]
The reason for the negative sign in front of the internal force is that we view the internal forces as exerted \emph{by} the mass distribution. Thus, the forces acting \emph{on} the mass distribution appear with a negative sign.

Let $\sigma\in \mathfrak{S}_\kappa$ be a stress representing the internal forces, that is,
\[
f_I=j^{r*}(\sigma)=\sigma\circ j^r.
\]
Typically, $\sigma$ will be determined by a constitutive relation. The total force acting on a body is
\[
f_T=f_E-j^{r*}(\sigma).
\]
Note that when the total force vanishes (i.e., in static equilibrium), 
the stress $\sigma$ represents the external force.

We further note that when the ambient space is Euclidean, (see Truesdell \cite{Tru91})  one formulates the dynamic laws in terms of a resultant force, a notion that has no counterpart for manifolds. 
It is possible, in the case of Euclidean spaces, to formulate the law for the external forces only because the work of the stresses for ``uniform'' velocity fields vanishes.

Given body and space manifolds, $\B$ and $\S$,
a \emph{motion} of the body is a smooth curve in configuration space,
\[
\vp:I\to \Q,
\]
where $I\subset\R$ is an interval. 
The \emph{velocity} associated with the motion $\vp$ is a map $V:I\to \vp^*T\Q$ defined by
\[
V_t=\left.\deriv{\vp}{t}\right|_t.
\]
For every $t$, $V_t$ is a vector field along $\vp(t)$. 
Given a time-dependent family of metrics $\{\GG(t)\}$ and a connection $\nabla^\Q$ for $\Q$, the momentum, $P:I\to \vp^*T^*\Q$, is the dual pairing of the velocity under the (time-dependent) metric $\GG(t)$, 
\[
P_t=\GG_{\vp(t)}(V_t,\cdot).
\]

The connection $\nabla^\Q$ on $T\Q$ induces a connection $\nabla^{\Q^*}$ on $T^*\Q$ by Leibniz' rule,
\[
(\nabla^{\Q^*}_\xi\Phi)(\eta) = \xi(\Phi(\eta)) - \Phi(\nabla^\Q_\xi\eta),
\qquad
\xi,\eta\in\Gamma(T\Q), 
\quad
\Phi\in\Gamma(T^*\Q).
\]
The \emph{inertial force}, i.e., the left-hand side of Newton's equation is given by
\[
\frac{DP}{dt}=(\nabla^{Q^*}_V P)_t.
\]
If $\GG$ is time-independent and $\nabla^\Q$ is metrically-consistent relative to $\GG$, then Newton's "$ma$" is recovered, namely,
\[
\frac{DP}{dt}= \GG_{\vp(t)}(A_t,\cdot) .
\]
where the \emph{acceleration} $A:I\to\vp^*T\Q$ is defined by $A_t=(\nabla^\Q_VV)_t$.

We now present the law of motion:
Let $\GG$ and $\nabla^\Q$ be as before, and $\vp:I\to\Q$ be a motion of $\B$ in $\S$. Assume that at  time $t\in I$, $\vp(t)$ is subject to a force $f_T=f_E-j^{r*}(\sigma)\in T^*_{\vp(t)}\Q$. Then $\vp(t)$ satisfies
the law of motion
\beq\label{lawofmotion}
\left.\frac{DP}{dt}\right|_t = f_T.
\eeq
Equation \eqref{lawofmotion} is a physical law relating the total force to the rate of change of the momentum; it is not a differential equation. Turning this physical law into a differential equation for the motion requires constitutive assumptions.

\subsection{Metric and connection for $\Q$}
\label{sec:metric_and_connection}

The constructions of metrics and connections on infinite-dimensional manifolds is far more involved than in the finite-dimensional case. Since the configuration space $\Q$ is infinite-dimensional, it lacks a partition of unity, and, it is not a-priori  clear  that there exist (globally defined) metrics and connections
for $\Q$. In this section, we follow Eliasson \cite{Eli67} and Palais \cite{Pal68}: 
we define a metric and a connection for $\Q$ using a canonical construction.

We start by noting that we may view $C^1(\B,T\S)$ as a vector bundle over $C^1(\B,\S)$. For every $f\in C^1(\B,\S)$, 
\[
(C^1(\B,T\S))_f=\{\eta\in C^1(\B,T\S)\,|\, \pi_\S\circ \eta=f   \}.
\]
Moreover, there exists a canonical vector bundle isomorphism 
\beq\label{caniso}
 C^1(\B,T\S)\simeq TC^1(\B,\S)
\eeq
which identifies every $\eta\in C^1(\B,T\S)$ with a tangent vector at $\pi_\S\circ \eta$, that is $\eta\in C^1((\pi_\S\circ\eta)^*T\S)\simeq T_{\pi_\S\circ\eta}C^1(\B,\S)$.

Assume a metric $g$ on $\S$ and a positive, non-vanishing $d$-form $\theta$ on $\B$, which we call the \emph{mass form}.  Using the isomorphism \eqref{caniso}, we define a metric $\GG$ on $\Q$ by 
\beq\label{metricdef}
\GG_\kappa(v,w) = \int_\B g_{\kappa(\cdot)}(v,w)\,\theta,
\eeq
where on the left-hand side we view $v$ and $w$ as elements of $T_\kappa\Q$, and on the right-hand side we view them as elements of $C^r(\kappa^*T\S)$.

The mass density of $\B$ is incorporated in the mass form $\theta$. Locally,
\[
\theta=\rho \, dX^1\wedge\cdots\wedge dX^d,
\] 
where $\rho:\B\to \R_+$ is a mass density function. 
In general, (e.g., for growing bodies), $\rho$ may be time-dependent, inducing a family of time-dependent metrics $\{\GG(t)\}_{t\in I}$ on $\Q$. In cases where $\B$ is endowed with a Riemannian metric $G$, it is natural to define the mass density $\rho$ to be the density of $\theta$ with respect to the Riemannian volume form, i.e.,
\[
\theta=\rho \, \sqrt{\det G} \, dX^1\wedge\cdots\wedge dX^d.
\]

Even more generally, one might consider a metric on $\Q$ of the form 
\[
\GG_\kappa(v,w) = \int_\B g_{\kappa(\cdot)}(v,w)\,\theta+\int_{\partial\B} g_{\kappa(\cdot)}(v,w)\,\theta_\partial,
\]
where $\theta_\partial$ is a surface form on $\partial\B$.
Metrics of this form are relevant to surface dynamics. 
In this paper we consider metrics of the form \eqref{metricdef}, i.e., metrics not having singular boundary contributions. The connection presented below turns out to be metrically-consistent with metrics of that form.

Following Eliasson \cite{Eli67}, we construct a connection for $T\Q$.
We start by defining the notion of connection maps. Let $M$ be a (possibly infinite-dimensional) manifold modeled on a banach space $\widetilde{M}$, and let $\pi_E:E\to M$ be a vector bundle over $M$ with fibers isomorphic to a Banach space $\hat{E}$. An element $e\in E$ is represented locally by a pair $(x,\xi)$, with $x\in \widetilde{M}$ and $\xi\in\hat{E}$. Likewise, an element of the tangent bundle $TE$ of $E$  is represented by a quadruple $(x,\xi,y,\eta)$ with $x,y\in\widetilde{M}$ and $\xi,\eta\in\hat{E}$. 

\begin{definition}
A \emph{connection map} for $E$ is a mapping $K:TE\to E$, which in every coordinate system has a local representative 
\[
\tilde{K}:\widetilde{M}\times \hat{E}\times\widetilde{M}\times\hat{E}\to \widetilde{M}\times\hat{E}
\]
of the  form
\[
\tilde{K}(x,\xi,y,\eta)=(x,\eta+\Gamma(x)(y,\xi)),
\]
 where $\Gamma(x):\tilde{M}\times \hat{E}\to \hat{E}$ is a bilinear transformation called the local connector of $K$ at $x$ (which should not be confused with our use of  $\Gamma$ to denote spaces of sections). 
\end{definition}

In the particular case where $M$ is finite-dimensional and $E=TM$, the local connector $\Gamma$ is given by  Christoffel symbols, 
\[
\Gamma(x)(v^ie_i,w^je_j)=\Gamma^k_{ij}(x)v^iw^j \, e_k.
\] 

Given a connection map $K$, one can define a connection $\nabla$ on $E$ in the following way:
Given a section $\xi\in\Gamma(E)$, set $\nabla\xi=K\circ T\xi\in\Gamma(\Hom(TM,E))$. That is, for $p\in M$ and $w\in T_pM$ 
\[
(\nabla_w\xi)_p=K(T\xi(w))\in E_p.
\]
If a section $\xi$ is represented by $\tilde{\xi}:\tilde{M}\to \hat{E}$, that is, locally $\xi(x)=(x,\tilde{\xi}(x))$ then a simple computation gives
\[
\nabla_w\xi(x)=(x,D\tilde{\xi}(x)(w)+\Gamma(x)(w,\tilde{\xi})).
\]

Turning back to the problem at hand, take $E=T\S$ and assume that a connection map $K^\S:T^2\S\to T\S$  is given, with the corresponding connection denoted by $\nabla^S$. One can then show (see \cite{Eli67} for details) that $K^S$ induces a connection map  
\[
C^1(K^\S):T^2C^1(\B,\S)\simeq C^1(\B,T^2\S)\to  TC^1(\B,\S)\simeq C^1(\B,T\S)
\]
defined by
\[
C^1(K^\S)(A)=K^\S\circ A,\quad A\in C^1(\B,T^2\S).
\]
Denote the restriction of $C^1(K^S)$ to $\Q$ (which is an open subset of $C^1(\B,\S)$) by $K^\Q$, and the corresponding connection $\nabla^\Q$.
For a section $\xi\in\Gamma(T\Q)$, a configuration $\kappa\in\Q$ and a tangent vector $w\in T_\kappa\Q$,
\beq
\label{conformula}
(\nabla^\Q_w\xi)_\kappa=(K^\Q\circ (T\xi)_\kappa)(w)=K^\S \circ\left((T\xi)_\kappa(w) \right).
\eeq
Note that on the right-hand side, $\left((T\xi)_\kappa(w) \right):\B\to T^2\S$ and $K^\S:T^2\S\to T\S$, hence, we obtain indeed a map $\B\to T\S$, i.e., an element of $T\Q$.

Since $\S$ is endowed with a metric $g$, there is a natural choice for $\nabla^\S$---the  Levi-Civita connection. One can show that in this case, $\nabla^\Q$ is metrically-consistent with respect to $\GG(t)$ for every $t\in I$.

Next, we derive for later use a local expression for the inertia term $DP/dt$, using  the metric and the connection defined above. Local coordinate systems for $\Q$ and $T\Q$ are given in terms of differential equations for the exponential map and Jacobi field respectively (see Eliasson \cite{Eli67}) and therefore cannot be given explicitly in the general case. The advantage of working with a connection map $K^\Q$, however, is that the covariant derivative can be calculated pointwise (in $\B$). We can therefore derive explicit expressions for the acceleration in coordinate neighborhoods  of $\B$ and $\S$.

Let $\vp:I\to\Q$ be a motion and let $V=\frac{d\vp}{dt}:I\to \vp^*T\Q$ be its velocity.  The acceleration $A:I\to \vp^*T\Q$ is given by 
\[
A_t=(\nabla^\Q_VV)_t=K^\S\circ (TV(\partial_t))_t.
\]
Let $(X^1,\dots,X^d)$ and $(x^1,\dots,x^m)$ be coordinate systems for $\B$ and $\S$ respectively.  If $\vp$ is represented by a vector of functions $\vp^i:I\times\R^d\to\R$, $1\leq i\leq m$,  then $V$ has a local representation $V^i=\partial\vp^i/\partial t$; for $t\in I$ and $p\in\B$
\[
V_t(p)=\pd{\vp^i}{t}(t,p)\partial_{x^i}.
\]
It follows that $TV(\partial_t)(t,p)\in T^2_{v_t(p)}\S$ is represented locally by 
\[
\brk{\vp^i(t,p),V^i(t,p),V^i(t,p), \pd{V^i}{t}(t,p)}.
\]
By the definition of the connection, $A_t(p)=K^\S(TV(\partial_t)(t,p))$ is represented locally by
\[
A^i(t,p)=\pd{V^i}{t}(t,p)+\Gamma^i_{jk}(\vp(t,p))\, V^j(t,p)\, V^k(t,p),
\]
where $\Gamma^i_{jk}$ are the Christoffel symbols of $\nabla^\S$.

As the inertial force is a one-form on $\Q$ (given by an integral functional), it is not possible to obtain a local expression as we did for the acceleration. However, as the momentum $P$ is given by 
 \[
P=\GG(V,\cdot)=\int_\B g(V,\cdot)\theta,
\]
we obtain (by the metricity of $\nabla^\Q$) that  the inertial force is given by
\beq\label{forceofinertia}
\frac{DP}{dt}=\int_\B g(A,\cdot)\theta+\int_\B g(V,\cdot)\dot{\theta}.
\eeq

It is possible to obtain a local representation of the integrands.
Suppose  that $g$ and $\theta$ are represented locally by 
\[
g=g_{ij}\, dx^i\otimes dx^j 
\Textand 
\theta=\rho \, dX^1\wedge\cdots\wedge dX^d.
\]
Then the integrands in \eqref{forceofinertia}
have the local form 
\[
\begin{split}
g(A_t,&\cdot)\theta(t)+g(V_t,\cdot)\dot{\theta}(t)=\\
&=g_{ij}\left(\rho\left( \frac{\partial^2\vp^i}{\partial t^2}+\Gamma^i_{lk}\frac{\partial\vp^l}{\partial t}   \frac{\partial\vp^k}{\partial t}\right) +\pd{\rho}{t} \deriv{\vp^i}{t}     \right)dx^j\otimes dX^1\wedge\cdots\wedge dX^d.
\end{split}
\]
Note that if $\theta$ does not depend on time, then the inertial force is the dual pairing of the acceleration via $\GG$, that is,
 \[
 \left.\frac{DP}{dt}\right|_t=\int_\B g(A,\cdot)\theta=(\flat^\GG A)_t,
 \]
 where $\flat^\GG:T\Q\to T^*\Q$ is the canonical map induced by $\GG$; unlike the finite-dimensional case, it is not an isomorphism.

\section{Constitutive theory}
\label{sc13}

As  mentioned in \secref{sec:force}, the total force at every configuration is decomposed into external and internal forces. In order to write the  equations of motion, we need to know the dependence of both internal and external forces on the configuration. 
Thus, the following are assumed to be given:
\begin{enumerate}
\item  A \emph{loading}, which is a one form  $\Phi:\Q\to T^*\Q$, assigning to every configuration $\kappa\in\Q$ an external force $\Phi_\kappa\in T_\kappa^*\Q$. 
\item A \emph{constitutive relation} $\Psi:\Q\to \mathfrak{S}$, assigning to every configuration $\kappa\in\Q$ a stress $\Psi_\kappa\in\mathfrak{S}_\kappa$. 
\end{enumerate}

The total force at a given configuration $\kappa\in\Q$ (which is an element of $T^*_\kappa\Q\simeq C^1(\ks)^*$) is  given by 
\[
(f_T)_{\kappa}=\Phi_\kappa-\Psi_\kappa\circ j^1.
\]
The total virtual power performed on a virtual velocity $w\in T_\kappa\Q$ is 
\[
(f_T)_{\kappa}(w)=\Phi_\kappa(w) - \Psi_\kappa(j^1w).
\]
Substituting the total force into \eqref{lawofmotion}, we obtain the equation of motion
\begin{equation}\label{dynamiceq}
\frac{DP}{dt}(w)=\Phi_{\vp(t)}(w)-\Psi_{\vp(t)}(j^1w),\quad \forall t\in I,\, w\in T_{\vp(t)}\Q .
\end{equation}

Generally, the constitutive relation and the loading may be singular, in which case \eqref{dynamiceq} may not have a local differential form. 
In the smooth case, where the external loading $\Phi$ is represented by a body force density $b$ and a surface force density $t$, and the constitutive relation $\Psi$ yields a stress that is represented by a variational stress density $S$, we obtain 
\beq\label{varprinciple}
\begin{split}
\int_\B g_{\vp(t)}(A_t,w)\theta + \int_\B g_{\vp(t)}(V_t,w)\dot{\theta} &= 
\int_\B b_{\vp(t)}(w) + \int_\B \divergence S_{\vp(t)}(w) \\
&+\int_{\partial\B}t_{\vp(t)}(w) - \int_{\partial\B}p_\sigma S_{\vp(t)}|_{\partial\B}(w)
\end{split}
\eeq
for every $t\in I$ and  $w\in T_\vp(t)\Q$, implying the following differential system:
\beq
g_{\vp(t)}(A_t,\cdot)\theta + g_{\vp(t)}(V_t,\cdot)\dot{\theta} = 
b_{\vp(t)} + \divergence S_{\vp(t)},
\label{eq:diff_gen}
\eeq
which is an identity between vector-valued forms in $\B$. The boundary conditions are
\[
t_{\vp(t)} = p_\sigma S_{\vp(t)}|_{\partial\B}.
\]
Equation \eqref{varprinciple} is a covariant version of  the principle of virtual work. 

A configuration $\kappa$ is an \emph{equilibrium configuration} if the total force vanishes, or in other words, if the constant motion $\vp(t)\equiv\kappa$ is a solution of the evolution equation \eqref{dynamiceq}.
The equilibrium condition yields  a boundary value problem,
\[
\begin{aligned}
&\divergence S_\kappa + b_\kappa= 0 &\qquad& \text{in $\B$} \\
& t_\kappa = p_\sigma S_\kappa|_{\partial\B} &\qquad& \text{on $\partial\B$}.
\end{aligned}
\]

\begin{remark}
The force-free equation $DP/dt=0$ may be dissipative if the mass density is time-dependent. 
If the mass density does not depend on time, the force-free equation is
\[
DP/dt=\GG(A,\cdot) = 0. 
\]
Its solution is a geodesic flow of $\B$ in $\S$. This is a covariant version of Newton's law of inertia in non-Euclidean continuum mechanics; every material element in a body free of both internal and external forces moves at constant speed along an $\S$-geodesics.
\end{remark}

\subsection{The generalized hyperelastic case}

A constitutive relation $\Psi$ for a variational stress density $S$ is said to be \emph{hyperelastic}  
if  $S$ is derived from an energy functional in the following way: Let 
\[
\W : J^1(\B,S) \to \R
\]
be an energy density function, and let $U:\Q\to\R$, given by
\[
U(\kappa)=\int_{\B} \W(j^1\kappa)\,\theta,
\]
be the corresponding energy functional. Then, $U$ induces a constitutive relation $(TU)_\kappa=\Psi_\kappa\circ j^1$ for every $\kappa\in\Q$.  The variational stress density $S$ of a hyperelastic system is given by 
\[
S_\kappa = \delta_{j^1\kappa}\W\otimes\theta
\]
where $\delta_{j^1\kappa}\W$ is the fiber derivative of $\W$ along $j^1\kappa$. That is, 
\[
\delta_{j^1\kappa}\W=\delta\W\circ j^1k,
\] 
and $\delta\W$ is the restriction of $T\W$ to the vertical sub-bundle of $TJ^1(\B,\S)$ (no derivatives in the $\B$ directions).

This definition of hyperelasticity is a generalization of the standard concept, in which it is assumed that the energy density only depends on the derivative of the configuration. As pointed out above, in a general geometric setting it is not possible to disassociate the derivative of a map at a point from the value of the map at that point.

In the absence of a loading, that is, in the case of a free motion, the equation of motion \eqref{dynamiceq} takes the form  
\beq\label{dyneqhe}
\frac{DP}{dt}=-(TU)_{\vp(t)}=-\int_\B \delta_{j^1\kappa}\W(\cdot)\theta. 
\eeq
As in classical mechanics we have conservation of energy which is due to the metricity of the connection $\nabla^\Q$ with respect to $\GG$:

\begin{proposition}
Let $\vp:I\to\Q$ be a free motion of a hyperelastic body, and suppose that the metric $\GG$ given by \eqref{metricdef} is time-independent . Define the kinetic energy $E_K:T\Q\to\R$ by $E_K(w)=\frac{1}{2}\GG(w,w)$. Then, 
\[
\frac{d}{dt}\left(E_K(V_t)+U(\vp(t)\right)=0.
\]  
\end{proposition}

\begin{proof}
By the chain rule 
\[
\frac{d}{dt}(U\circ\vp)(t)=(TU)_{\vp(t)}\circ \frac{d\vp}{dt}=(TU)_{\vp(t)}(V_t).
\]
As $\nabla^\Q$ is metric with respect to $\GG$ we have 
\[
\frac{d}{dt}(E_K(V_t))=\frac{1}{2}\frac{d}{dt}\GG(V_t,V_t)=\GG((\nabla^\Q_VV)_t,V)=\GG(A_t,V_t)=\frac{DP}{dt}(V_t).
\]
Hence, by \eqref{dyneqhe} 
\[
\frac{d}{dt}\left(E_K(V_t)+U(\vp(t)\right)=\frac{DP}{dt}(V_t)+(TU)_{\vp(t)}(V_t)=0\, .
\]
\end{proof}

Locally, $\W$ is represented by a function $\R^m\times\R^{d\times m}\to\R$, and for every $w\in T_\kappa\Q$
\[
S_\kappa(w^i,w^i_{,\alpha}) = (R_i w^i+S^\alpha_i w^i_{,\alpha}) \, dX^1\wedge\cdots\wedge dX^d,
\]
where
\beq\label{hyperstressformula}
R_i = \rho\,\pd{\W}{x^i}\,(j^1\kappa)
\textand 
S^\alpha_i = \rho\,\pd{\W}{x^i_{,\alpha}}(j^1\kappa).
\eeq

In the absence of external loadings, and with the metric and connection $\GG$ and $\nabla^\Q$ defined as in \secref{sec:metric_and_connection}, the equation of motion for $\vp:I\to\Q$ takes the following form: for every vector field $\xi:I\to T\Q$ along $\vp$, 
\begin{equation}\label{hyperdyn}
\int_\B g(A_t,\xi_t)\theta+\int_\B g(V_t,\xi_t)\dot{\theta} =
\int_B \divergence S_{\vp(t)}(\xi_t) - \int_{\partial\B}p_\sigma S_{\vp(t)}(\xi_t).
\end{equation}
The corresponding differential equation has the local form 
\begin{equation}\label{hyperdyncord}
g_{ij}\brk{\pdd{\vp^i}{t}+\Gamma_{lk}^i\pd{\vp^l}{t}\pd{\vp^k}{t} +
\frac{\dot{\rho}}{\rho}\pd{\vp^i}{t}} =
\frac{1}{\rho}\partial_\alpha\brk{\rho \pd{\W}{x^j_{,\alpha}}(j^1\vp)} - \pd{\W}{x^j}(j^1\vp),
\end{equation}
with boundary conditions
\[
\sum_\alpha(-1)^{\alpha-1}\pd{\W}{x^i_{,\alpha}}(j^1\vp)\,dX^1\wedge\cdots\wedge\widehat{dX^\alpha}\wedge\cdots\wedge dX^d=0 
\qquad \text{on}\,\partial\B.
\]

\section{A quadratic constitutive model}
\label{sec:MooneyRivlin}

In most applications, the body manifold $\B$ of an elastic medium has an intrinsic geometry---a Riemannian metric $G$---and the elastic energy density $\W(j^1_p\kappa)$ is a measure of the local strain: it measures the local distortion induced by the current configuration $\kappa\in\Q$ at the point $p\in\B$. Moreover, the Riemannian metric $G$ induces a natural (time-independent) volume form on $\B$ which we denote by $\Vol_G$. In coordinates, $G = G_{ij}\, dX^i\otimes dX^j$, and 
\[
\Vol_{G} = \sqrt{\det G}\,dX^1\wedge\cdots \wedge dX^d.
\]

A configuration $\kappa\in\Q$ induces on $\B$ a metric $\kappa^\star g$ measuring ``actual" distances and angles in $\B$ induced by its embedding in $\S$; this metric is known in continuum mechanics as the \emph{right Cauchy-Green deformation tensor}. In coordinates, the entries of $\kappa^\star g$ are
\[
(\kappa^\star g)_{\alpha\beta}= (\kappa^*g_{ij})\pd{\kappa^i}{X^\alpha}\pd{\kappa^j}{X^\beta},
\] 
where $\kappa^*g_{ij}(p) = g_{ij}(\kappa(p))$. 

A notational convention: 
we denote by $\kappa^* g$
the section of $\kappa^*(T^*\S\otimes T^*\S)$ obtained by pulling-back $g$ considered as a section of $T^*\S\otimes T^*\S$. This should not be confused with the closely related pullback of $g$ considered as a one-form on $\S$, involving composition with $d\kappa$, which we denote by $\kappa^\star g$.

The deviation $\epsilon = \kappa^\star g-G$ of the actual metric $\kappa^*g$ from the intrinsic metric $G$ at a point is a measure of local strain; it is known (up to a factor of $1/2$) as the \emph{Green-St Venant strain tensor}. The elastic energy density is a function of this strain, vanishing at $p\in\B$ if and only if $(\kappa^\star g)_p = G_p$.

The specific form of the energy density depends on the material under consideration. 
A generalization of Hooke's law assumes an elastic energy density that is quadratic in the strain,
\[
\W(j^1\kappa) = C^{\alpha\beta\gamma\delta} \epsilon_{\alpha\beta} \epsilon_{\gamma\delta},
\]
where $C^{\alpha\beta\gamma\delta}$ is called the \emph{elasticity tensor}. If material isotropy is assumed, then the $d^4$ components of $C^{\alpha\beta\gamma\delta}$ reduce to only two component, 
\[
C^{\alpha\beta\gamma\delta} = \lambda \G^{\alpha\beta} G^{\gamma\delta} + \frac{\mu}{2} \brk{G^{\alpha\gamma} G^{\beta\delta} + G^{\alpha\delta} G^{\beta\gamma}}.
\]
The parameters $\lambda$ and $\mu$ are known in the linearized three-dimensional Euclidean case as Lam\'e first and second coefficients.
For the particular case where $\lambda=0$ and $\mu=1$, the elastic energy reduces to
\[
\W(j^1\kappa)= \|\kappa^\star g-G\|^2,
\]
where the norm $\|\cdot\|$ is induced by the metric $G$.
In coordinates,
\[
\begin{split}
\W(j^1\kappa) &= G^{\alpha \gamma} G^{\beta \delta} ((\kappa^\star g)_{\alpha \beta} - G_{\alpha \beta})((\kappa^\star g)_{\delta\gamma} - G_{\delta\gamma} )\\
&= G^{\alpha \gamma} G^{\beta \delta} 
\brk{(\kappa^*g_{ij})\pd{\kappa^i}{X^\alpha}\pd{\kappa^j}{X^\beta} - G_{\alpha \beta}}
\brk{(\kappa^*g_{lk})\pd{\kappa^l}{X^\delta}\pd{\kappa^l}{X^\gamma} - G_{\delta\gamma}}.
\end{split}
\]
The derivatives of the energy density are
\beq
\label{verder1}
\pd{\W}{x^i_{,\alpha}}(j^1\kappa) = 4 G^{\alpha \gamma} G^{\beta \delta} 
(\kappa^*g_{ij})\pd{\kappa^j}{X^\beta} 
\brk{(\kappa^*g_{lk})\pd{\kappa^l}{X^\delta}\pd{\kappa^k}{X^\gamma} - G_{\delta\gamma}},
\eeq
and
\beq
\label{verder2}
\pd{\W}{x^m}(j^1\kappa) = 2 G^{\alpha \gamma} G^{\beta \delta} 
\brk{\kappa^*\pd{g_{ij}}{x^m}}\pd{\kappa^i}{X^\alpha}\pd{\kappa^j}{X^\beta} 
\brk{(\kappa^*g_{lk})\pd{\kappa^l}{X^\delta}\pd{\kappa^k}{X^\gamma} - G_{\delta\gamma}}.
\eeq
\begin{remark}
The $R_i$ terms in the variational stress density are non-zero since the metric of the ambient space  $g$ has non-trivial spatial derivatives. In conventional elasticity theories, the spatial metric is Euclidian and the $R$ term vanishes.
\end{remark}

We substitute these expression into Equation \eqref{hyperstressformula} for $R_i$ and $S_j^\alpha$, with $\rho = \sqrt{\det G}$, and then into Equation \eqref{eq:divS} for the coordinate representation of $\divergence S$, getting
\beq
\label{divhyper}
(\divergence S)_i = 
\pd{}{X^\alpha}\brk{\sqrt{\det G} \,\pd{\W}{x^i_{,\alpha}}} - \sqrt{\det G}\,\pd{\W}{x^i} .
\eeq

In summary, let $\kappa_0\in\Q$ be an initial configuration and let $v_0\in T_\kappa\Q$ be an initial velocity. 
Assume free boundary conditions. Then, the coordinate form of the equations of motion is 
\beq
\label{decase1}
\sqrt{\det G}\, (\vp^*g_{ij}) \brk{\pdd{\vp^j}{t}+\Gamma^j_{lk}\pd{\vp^l}{t}      
\pd{\vp^k}{t}} =  (\divergence S(\vp))_i,
\eeq
with boundary conditions
\[
\sum_\alpha(-1)^{\alpha-1}S^{\alpha}_i(\vp) \, dX^1\wedge\cdot\wedge \widehat{dX^\alpha}\wedge \cdots\wedge dX^d = 0 \quad \text{on} \,\partial\B,
\]
and initial conditions
\[
\vp_0=\kappa_0, \Textand  \dot{\vp}_0 = v_0.
\]

\section{An example}
\label{sec:example}

Let the body manifold $\B$ be a two-dimensional spherical annulus, with a coordinate system
\[
(R,\Theta) \in [\rmin,\rmax]\times [0,2\pi)
\]
and periodicity in the second coordinate; we take an annulus rather than a disc just in order to avoid the immaterial singularity of the polar coordinate system.

The body manifold is assumed to be endowed with an azimuthally-symmetric metric of the form
\[
G(R,\Theta) = \mymat{1 & 0 \\ 0 & \Phi^2(R)}.
\]

For example, the choice of 
\[
\Phi(R) = \frac{\sin\sqrt{K}R}{\sqrt{K}}
\]
with $K>0$ corresponds to a spherical cap of constant Gaussian curvature $K$, whereas the choice of 
\beq
\Phi(R) = \frac{\sinh\sqrt{K}R}{\sqrt{K}}
\label{eq:sphere}
\eeq
corresponds to a hyperbolic cap of constant Gaussian curvature $(-K)$.

The space manifold is a two-dimensional disc. Let 
\[
(r,\theta) \in [0,\infty)\times [0,2\pi)
\]
be a coordinate system for $\S$, with periodicity in the second coordinate. 
The space manifold is also endowed with an azimuthally-symmetric metric of the form
\[
g(r,\theta) = \mymat{1 & 0 \\ 0 & \phi^2(r)}.
\]
If $(\S,g)$ has positive Gaussian curvature then the range of $r$ must be bounded.

The non-vanishing Christoffel symbols associated with the metric $G$ are 
\[
\Gamma^r_{\theta\theta}(r,\theta) = -\phi(r)\phi'(r)
\textand
\Gamma^\theta_{r\theta}(r,\theta) =  \phi'(r)/\phi(r).
\]

This setting encompasses a large family of elastic systems that have received much interest in recent years, such as spherical caps embedded in the plane, a hyperbolic disc embedded in the plane \cite{ESK09b} or a flat surface embedded on a sphere \cite{KSDM12}.  

All the aforementioned references treat only elastostatics.
The only exception we are aware of is the recent work of Aharoni et al. \cite{AKMMS16}, which studies the motion of a quasi-dimensional reactive gel confined within a thin gap between two non-planar plates (a curved version of a Hele-Shaw plate). This setting mimics a two-dimensional body moving in a non-flat two-dimensional space manifold.
The plates were curved such that the top part has a elliptic geometry and the bottom part has an hyperbolic geometry (\figref{fig:experiment1}).

\begin{figure}
\begin{center}
\includegraphics[height=1.8in]{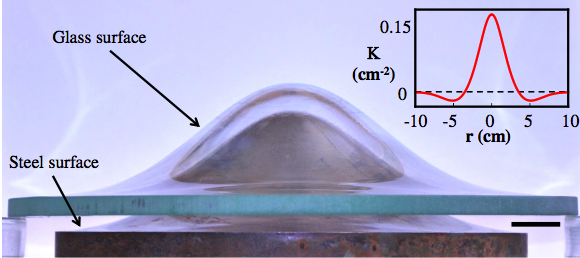}
\end{center}
\caption{The experimental setting in \cite{AKMMS16}: The inset displays the cell’s Gaussian curvature as a function of the radius.}
\label{fig:experiment1}
\end{figure}

This setup is immersed in a temperature-regulated water bath; by controlling the temperature,  the intrinsic curvature of the gel can be modified.  When the curvature of the body changes from hyperbolic to elliptic, the body migrates from the lower portion of the cell to the upper portion (\figref{fig:experiment2}). It should be noted that these experiments correspond to a damped regime, hence cannot be quantitatively compared to our computations below. Yet, unlike Hamiltonian formulations, our approach can account for dissipation. 

\begin{figure}
\begin{center}
\includegraphics[height=2.4in]{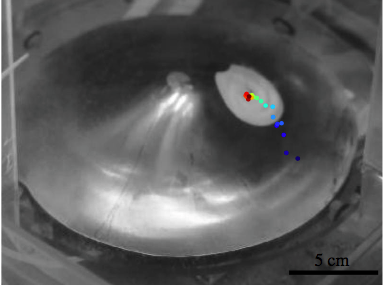}
\end{center}
\caption{The gel in its final, stable position. The dots trace out the trajectory of the gel, from start (blue) to end (red).}
\label{fig:experiment2}
\end{figure}

Consider now a time-dependent configuration preserving the azimuthal symmetry of the system, 
\[
\vp^r(R,\Theta,t) = f(R,t)
\Textand
\vp^\theta(R,\Theta,t) = \Theta,
\]
for some function  $f:[\rmin,\rmax]\times I \to [0,\infty)$.

Substituting this \emph{ansatz} into \eqref{verder1} and \eqref{verder2}, we obtain 
\[
\begin{split}
\pd{\W}{x^r_{,R}}(j^1\vp) &= 4 G^{RR} G^{RR} 
(\vp^*g_{rr})\pd{\vp^r}{R} 
\brk{(\vp^*g_{rr})\pd{\vp^r}{R}\pd{\vp^r}{R} - G_{RR}} = 4 f' \brk{f'^2 - 1},
\end{split}
\]
where $f' = f'(R,t)$ denotes derivation with respect to $R$, 
\[
\begin{split}
\pd{\W}{x^\theta_{,\Theta}}(j^1\vp) &= 4 G^{\Theta \Theta} G^{\Theta\Theta} 
(\vp^*g_{\theta\theta})\pd{\vp^\theta}{\Theta} 
\brk{(\vp^*g_{\theta\theta})\pd{\vp^\theta}{\Theta}\pd{\vp^\theta}{\Theta} - G_{\Theta\Theta}} \\
&= 4\frac{(\phi\circ f)^2}{\Phi^2} \brk{\frac{(\phi\circ f)^2}{\Phi^2}  - 1} ,
\end{split}
\]
and
\[
\begin{split}
\pd{\W}{x^r}(j^1\vp) &= 2 G^{\Theta\Theta} G^{\Theta\Theta} 
\brk{\vp^*\pd{g_{\theta\theta}}{r}}\pd{\vp^\theta}{\Theta}\pd{\vp^\theta}{\Theta} 
\brk{(\vp^*g_{\theta\theta})\pd{\vp^\theta}{\Theta}\pd{\vp^\theta}{\Theta} - G_{\Theta\Theta}} \\
&= \frac{4(\phi\circ f)(\phi'\circ f)}{\Phi^2} \brk{\frac{(\phi\circ f)^2}{\Phi^2}  - 1}.
\end{split}
\]
All the other derivatives are zero.

Substituting into \eqref{divhyper} we obtain the divergence of the stress,
\[
(\divergence S)_r = 
\pd{}{R}\brk{\Phi \,\pd{\W}{x^r_{,R}}} - \phi\,\pd{\W}{x^r} .
\]
\[
(\divergence S)_\theta = \pd{}{\Theta}\brk{\Phi \,\pd{\W}{x^\theta_{,\Theta}}} = 0,
\]
as well as the boundary term
\[
(p_\sigma S)_r = \pd{\W}{x^r_{,R}}.
\]

If $\phi = \Phi$, i.e., the two manifolds are compatible, then $\vp^r(R,t) = R$ is a stationary solution of this boundary value problem corresponding to an isometric embedding of $\B$ into $\S$. Otherwise, no isometric embedding exists, and the stress is non-zero even in the absence of external loads.

Finally, substituting into the equations of motion \eqref{decase1} for $i=r$, we obtain
\beq
\pdd{f}{t} = \frac{1}{\Phi} \pd{}{R}\brk{\Phi \,\pd{\W}{x^r_{,R}}} \Raz{-} \pd{\W}{x^r}.
\label{the_wave_eq}
\eeq
Note that the acceleration in the radial direction is simply a second derivative because we chose semi-geodesic coordinates for both $\B$ and $\S$. 
The boundary conditions are
\[
\pd{\W}{x^r_{,R}} = 0,
\]
which reduce to
\[
f'(\rmin,t) = f'(\rmax,t) = 1.
\]

We next present a particular calculation for a spherical annulus embedded in a sphere.
The radial coordinate $R$ of body manifold $\B$ range from $\rmin=0.2$ to $\rmax=1.0$. The metric is of the form \eqref{eq:sphere} with positive Gaussian curvature $K=2$. The space manifold $\S$ is a sphere with Gaussian curvature $k=0.5$. The curvature discrepancy implies that the body manifold cannot be embedded in the space manifold without stretching its outer part.

In \figref{fig:1} we show the equilibrium configuration. The top figure displays an isometric embedding of the body manifold in three-dimensional Euclidean space. The bottom figure displays an isometric embedding of its equilibrium configuration. Note that while the distance between the outer and inner boundaries in the body manifold is $0.8$, the actual distance between those boundaries at equilibrium is $0.716$. The effect of embedding a spherical annulus on a sphere of lesser curvature is compression.

\begin{figure}[h]
\begin{center}
\includegraphics[height=2.5in]{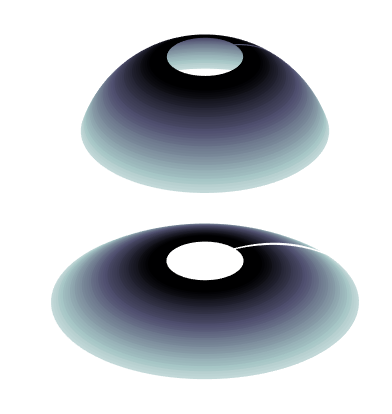}
\end{center}
\caption{Top: Isometric immersion of a spherical annulus of Gaussian curvature $K=2$ in Euclidean space. Bottom: Equilibrium configuration of that same annulus on a sphere of Gaussian radius $k=0.5$.  The stress at equilibrium is non-zero, exhibiting compressive forces in the outer parts.}
\label{fig:1}
\end{figure}

Next, we perturb the equilibrium configuration and solve numerically the nonlinear wave equation \eqref{the_wave_eq}. \figref{fig:2} displays the time evolution of the distance between the inner and outer boundaries over $10$ time units. As expected, we obtain oscillations. Note the multimodal nature of those oscillations, as expected from a nonlinear wave equation. 

\begin{figure}[h]
\begin{center}
\includegraphics[height=2.8in]{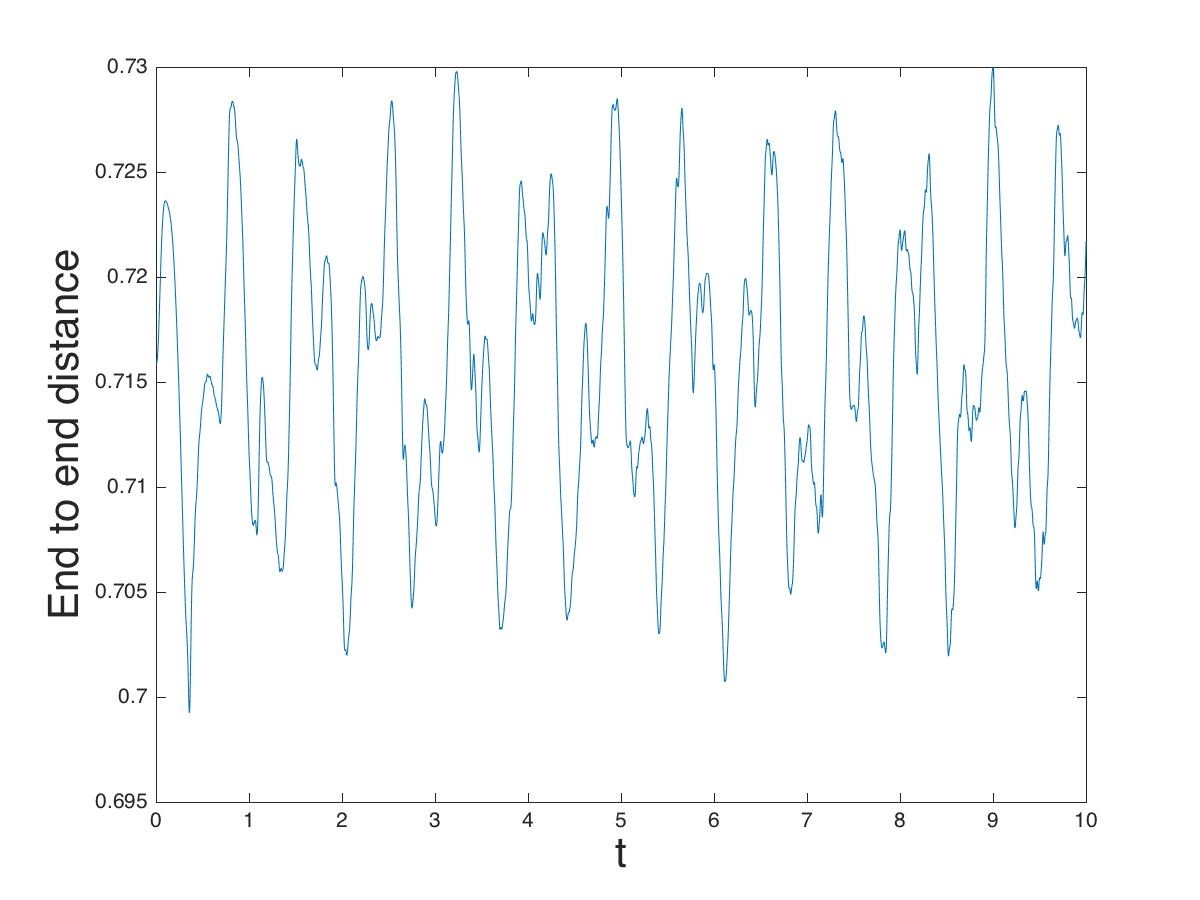}
\end{center}
\caption{Time evolution of the inner-boundary to outer-boundary distance for the system described in Figure~\ref{fig:1}.}
\label{fig:2}
\end{figure}

\begin{acknowledgements}
We are grateful to the author of \cite{AKMMS16} for letting us report their, yet, unpublished results.\end{acknowledgements}

\bibliographystyle{spmpsci}      
\bibliography{/Users/raz/Dropbox/tex/Refs/MyBibs}

\end{document}